\documentclass[reqno,12pt]{amsart}%
\usepackage{amsfonts}
\usepackage{verbatim}
\usepackage{xcolor}
\usepackage{amsmath}
\usepackage{amssymb}
\usepackage{graphicx}
\usepackage{accents}%
\providecommand{\U}[1]{\protect\rule{.1in}{.1in}}
\newtheorem{theorem}{Theorem}[section]

\newtheorem{corollary}[theorem]{Corollary}

\newtheorem{remark}[theorem]{Remark}
\newtheorem{lemma}[theorem]{Lemma}

\numberwithin{equation}{section}
\begin{document}
\title[KdV equation]{The effect of a positive bound state on the KdV solution. A case study.}
\author{Alexei Rybkin}
\address{Department of Mathematics and Statistics, University of Alaska Fairbanks, PO
Box 756660, Fairbanks, AK 99775}
\email{arybkin@alaska.edu}
\thanks{The author is supported in part by the NSF grant DMS 1716975.}
\date{November, 2018}
\subjclass{34B20, 37K15, 47B35}
\keywords{KdV equation, embedded eigenvalues, Wigner-von Neumann potentials}

\begin{abstract}
We consider a slowly decaying oscillatory potential such that the
corresponding 1D Schr\"{o}dinger operator has a positive eigenvalue embedded
into the absolutely continuous spectrum. This potential does not fall into a
known class of initial data for which the Cauchy problem for the Korteweg-de
Vries (KdV)\ equation can be solved by the inverse scattering transform. We
nevertheless show that the KdV equation with our potential does admit a closed
form classical solution in terms of Hankel operators. Comparing with rapidly
decaying initial data our solution gains a new term responsible for the
positive eigenvalue. To some extend this term resembles a positon (singular)
solution but remains bounded. Our approach is based upon certain limiting
arguments and techniques of Hankel operators.

\end{abstract}
\maketitle
\tableofcontents

\section{Introduction}

We are concerned with the initial value problem for the Korteweg-de Vries
(KdV) equation%
\begin{equation}%
\begin{array}
[c]{cc}%
\partial_{t}u-6u\partial_{x}u+\partial_{x}^{3}u=0,\ \ \  & -\infty
<x<\infty,\ t\geq0,\\
u\left(  x,0\right)  =q\left(  x\right)  . &
\end{array}
\label{KdV}%
\end{equation}
As is a well-known, for smooth rapidly decaying $q$'s (\ref{KdV}) was solved
in closed form in the short 1967 paper \cite{GGKM67} by
Gardner-Greene-Kruskal-Miura (GGKM). This seminal paper introduces what we now
call the\emph{ inverse scattering transform }(IST). Conceptually, it is
similar to the Fourier transform (see e.g. the classical books \cite{AC91},
\cite{NPZ}) but based on the inverse scattering theory  for the
Schr\"{o}dinger operator%
\begin{equation}
\mathbb{L}_{q}=-\partial_{x}^{2}+q(x)\text{ on }L^{2}\left(  \mathbb{R}%
\right)  .\label{schrodinger}%
\end{equation}
Moreover, the solution $q\left(  x,t\right)  $ to (\ref{KdV}) for each $t>0$
can be obtained by the formula%
\begin{equation}
u\left(  x,t\right)  =-2\partial_{x}^{2}\log\tau\left(  x,t\right)
,\label{tau}%
\end{equation}
where $\tau$ is the so-called \emph{Hirota tau-function} introduced in
\cite{Hirota71} which admits an explicit representation in terms of the
scattering data of the pair $\left(  \mathbb{L}_{q},\mathbb{L}_{0}\right)  $.
The solution has a relatively simple and by now well understood wave structure
of running (finitely many) solitons accompanied by radiation of decaying waves
(see e.g. Grunert-Teschl \cite{GT09} for a streamlined modern exposition). In
about 1973, the IST was extended to $q$'s rapidly approaching different
constants $q_{\pm}$ as $x\rightarrow\pm\infty$ (step initial profile). It
appeared first in the physical literature \cite{Gurevich73} and was rigorously
treated in 1976 by Hruslov\footnote{Also transcripted as Khruslov.}
\cite{Hruslov76}. The formula (\ref{tau}) is also available in this case with
an explicit representation of the \ tau-function in terms of certain
scattering data. We refer to our recent \cite{GruRybSIMA15} and
\cite{RybPhysD2018} where (\ref{tau}) is extended to essentially arbitrary
$q$'s with a rapid decay only at $+\infty$. The main feature of such initial
profiles is infinite sequence of solitons emitted by the initial step. Note
that a complete rigorous investigation of all other asymptotic regimes and
their generalizations was done only recently by Teschl with his collaborators
(see e.g. \cite{TeschlRarefaction16}, \cite{Egorovaetal13},
\cite{TeschlShock2016}).

Another equally important and explicitly solvable case is when $q$ is
periodic. The periodic IST\ is quite different from the GGKM\ one and is
actually the \emph{inverse spectral transform} (also abbreviated as IST) since
it relies on the Floquet theory for $\mathbb{L}_{q}$ and analysis of Riemann
surfaces and hence is much more complex than the rapidly decaying case. The
solution $u\left(  x,t\right)  $ is given essentially by the same formula
(\ref{tau}), frequently referred to as the \emph{Its-Matveev formula}
\cite{ItsMat75} (see also \cite{DubMatNov76} by Dubrovin-Matveev-Novikov and
the 2003 Gesztesy-Holden book \cite{GesHold03} where a complete history is
given), but $\tau$ is a multidimensional\footnote{Infinite dimensional in
general.} theta-function of real hyperelliptic algebraic curves explicitly
computed in terms of spectral data of the associated Dirichlet problem for
$\mathbb{L}_{q}$. It is therefore very different from the rapidly decaying
case. The main feature of a periodic solution is its quasi-periodicity in time
$t$.

We have outlined two main classes of initial data $q$ in (\ref{KdV}) for which
a suitable form of the IST\ was found during the initial boom followed by
\cite{GGKM67}. Such progress was possible due to well-developed inverse
scattering/spectral theories for the underlying potentials $q$. However, while
we have proven \cite{GruRybSIMA15} that no decay at $-\infty$ is required to
do the IST but slower than $x^{-2}$ decay at $+\infty$ results in serious
complications. The main issue here is that the classical inverse scattering
theory, the foundation for the IST, has not been extended beyond short-range
potentials, i.e. $q\left(  x\right)  =O\left(  \left\vert x\right\vert
^{-2-\varepsilon}\right)  $, $x\rightarrow\pm\infty$. We emphasize that during
the boom in scattering theory there was a number of results on (direct)
scattering/spectral theory for a variety of long-range potentials but the
inverse scattering theory is a different matter. It was shown in 1982
\cite{ADM81} that the short-range scattering data no longer determine the
potential uniquely even in the case when  $q\left(  x\right)  =O\left(
x^{-2}\right)  $ and it is not merely a technical issue of adding some extra
data. The problem appears to be open even for $L^{1}$ potentials (see
Aktuson-Klaus \cite{AK01}) for which all scattering quantities are
well-defined but may exhibit an erratic behavior at zero energy which is
notoriously difficult to analyze and classify. Besides, a possible infinite
negative spectrum begets an infinite sequence of norming constants which can
be arbitrary. Consequently, it is even unclear how to state a (well-posed)
Riemann-Hilbert problem which would solve the inverse scattering problem. Once
we leave $L^{1}$ then infinite embedded singular spectrum may appear leaving
no hope to figure out what true scattering data might be. We note that any
attempt to try the inverse spectral transform instead runs into equally
difficult problems (see, e.g. our \cite{RybJMP08} and the literature cited
therein) as spectral data evolve in time under the KdV flows by a simple law
essentially only for the so-called finite gap potentials. In addition, it
makes sense to find a suitable IST for (\ref{KdV}) if (\ref{KdV}) is actually
well-posed. The seminal 1993 Bourgain's paper \cite{Bourgain93} says that
(\ref{KdV}) is well-posed if $q$ is in $L^{2}$ and not much better result
should be expected regarding the decay at $+\infty$.

In the current paper we look into a specific representative of the important
class of continuous potentials asymptotically behaving like
\begin{equation}
q\left(  x\right)  =\left(  c/x\right)  \sin2x+O\left(  x^{-2}\right)
,\ \ \ x\rightarrow\pm\infty.\label{WvN pots}%
\end{equation}
In the half line context such potentials\footnote{In fact, for 3D radially
symmetric potentials.} first appeared in 1929 in the famous paper
\cite{WvN1929} by Wigner-von Neumann where they explicitly constructed a
potential of type (\ref{WvN pots}) with $c=-8$ which supports bound state $+1$
embedded in the absolutely continuous spectrum. Note that in general any $q$
of type (\ref{WvN pots}) with $\left\vert c\right\vert >2$ may support a bound
state $+1$ which is extremely unstable and turns into the so-called
\emph{Wigner-von Neumann resonance} under a small perturbation. If $\left\vert
c\right\vert >1/\sqrt{2}$ then the negative spectrum (necessarily discrete) of
$\mathbb{L}_{q}$ is infinite in general \cite{Klaus82}. While there is a very
extensive literature on potentials of type (\ref{WvN pots}) (commonly referred
to as Wigner-von Neumann type potentials) but, as Matveev puts it in
\cite{MatveevOpenProblems},\ "The related inverse scattering problem is not
yet solved and the study of the related large times evolution is a very
challenging problem". Observe that since any Wigner-von Neumann potential is
clearly in $L^{2}$, the Bourgain Theorem \cite{Bourgain93} guarantees
well-posedness of (\ref{KdV}) and the good open problem is if we can solve it
by a suitable IST. Our goal here is to investigate a specific case of
(\ref{WvN pots}) which can be done by the IST. Namely, we consider an even
potential $Q\left(  x\right)  $ defined for $x\geq0$ by%
\[
Q\left(  x\right)  =-2\frac{d^{2}}{dx^{2}}\log\left(  1+\rho x-\frac{\rho}%
{2}\sin2x\right)  ,
\]
where $\rho$ is an arbitrary positive constant. One can easily check that
$Q\,$\ is continuos and behaves like (\ref{WvN pots}) with $c=-4$. The main
feature of $Q$ is that $\mathbb{L}_{Q}$ admits an explicit spectral analysis
and consequently the scattering problem for the pair $\left(  \mathbb{L}%
_{Q},\mathbb{L}_{0}\right)  $ can also be solved explicitly. In particular,
$+1$ is a positive bound state of $\mathbb{L}_{Q}$ but its negative spectrum
consists of just one bound state. We show that for (\ref{KdV}) with initial
data $Q$ the tau-function in (\ref{tau}) can be explicitly calculated. The
formula however is expressed in the language of Hankel operators (which is not
commonly used in integrable systems) and we have to postpone it till Section
\ref{main results}. We only mention here that, comparing to the short range
case, the tau-function $\tau$ gains an extra factor responsible for the
positive bound state. Unfortunately, we were unable to find the IST even in
this case but we able to detour it by means of suitable limiting arguments.
Our limiting arguments are based on certain short range approximations of $Q$
combined with techniques of Hankel operators developed in our
\cite{GruRybSIMA15}.

The reader will see that our approach is not restricted to just one initial
condition and should work for a whole class of initial data (at least
\cite{NovikovKhenkin84} gives some hopes). We however do not make an attempt
to be more general for two reasons. First of call, our consideration would
complicate a great deal due to numerous extra technicalities. But the main
reason is that the scattering theory, the backbone of our approach, is not
developed well enough outside of short-range potentials. (At least not to our
satisfaction). For instance, there are only some results on regularity
properties of scattering data for Wigner-von Neumann type potentials (see
\cite{Klaus91}) but almost nothing is known about their small energy behavior.
The latter was posed as an open question in \cite{Klaus91} but, to the best of
our knowledge, there has been no progress in this direction since then. This
is a major impediment to our approach as it requires a careful control of the
scattering matrix at all energy regimes.

\section{Our analytic tools\label{Hankel}}

To translate our problem into the language of Hankel operators some common
definitions and facts are in order \cite{Nik2002}, \cite{Peller2003}.

\subsection{Riesz projections}

Recall, that a function $f$ analytic in the upper half plane $\mathbb{C}^{\pm
}:=\left\{  z|\pm\operatorname{Im}z>0\right\}  $ is in the \emph{Hardy space}
$H_{\pm}^{2}$ of $\mathbb{C}^{\pm}$ if%
\[
\sup_{h>0}\int_{\mathbb{R}\pm ih}\left\vert f\left(  z\right)  \right\vert
^{2}\left\vert dz\right\vert <\infty.
\]

It is a fundamental fact of the theory of Hardy spaces that any $f\in H_{\pm
}^{2}$ has non-tangential boundary values $f\left(  x\pm i0\right)  $ for
almost every (a.e.) $x\in\mathbb{R}$ and $H_{\pm}^{2}$ are subspaces of
$L^{2}:=L^{2}\left(  \mathbb{R}\right)  $. Thus, $H_{\pm}^{2}$ are Hilbert
spaces with the inner product induced from $L^{2}$:
\[
\langle f,g\rangle_{H_{\pm}^{2}}=\langle f,g\rangle_{L^{2}}=\left\langle
f,g\right\rangle =\frac{1}{2\pi}\int_{-\infty}^{\infty}f\left(  x\right)
\bar{g}\left(  x\right)  dx.
\]
It is well-known that $L^{2}=H_{+}^{2}\oplus H_{-}^{2},$ the \emph{orthogonal
(Riesz) projection} $\mathbb{P}_{\pm}$ onto $H_{\pm}^{2}$ being given by%
\begin{align}
(\mathbb{P}_{\pm}f)(x) &  =\pm\frac{1}{2\pi i}\lim_{\varepsilon\rightarrow
0+}\int_{-\infty}^{\infty}\frac{f(s)ds}{s-(x\pm i\varepsilon)}%
\label{Riesz proj}\\
&  =\pm\frac{1}{2\pi i}\int_{-\infty}^{\infty}\frac{f(s)ds}{s-(x\pm
i0)}.\nonumber
\end{align}
In what follows, we set $H_{+}^{2}=H^{2}$. Notice that for any $f\in H^{2}$%
\begin{equation}
\mathbb{P}_{-}\left(  \frac{1}{\cdot-\lambda}f\right)  =\frac{1}{\cdot
-\lambda}f(\lambda),\ \ \ \lambda\in\mathbb{C}^{+}.\label{P_}%
\end{equation}
Besides $H_{\pm}^{2}$, we will also use $H_{\pm}^{\infty}$, the algebra of
uniformly bounded in $\mathbb{C}^{\pm}$ functions.

\subsection{Reproducing kernels\label{rep kernels}}

Recall that, a given fixed $\lambda\in\mathbb{C}^{\pm}$ the function%
\begin{equation}
k_{\lambda}\left(  z\right)  :=\frac{i}{z-\overline{\lambda}},\ \ \ \lambda
\in\mathbb{C}^{\pm}\label{repro ker}%
\end{equation}
is called the\emph{ reproducing (or Cauchy-Szego) kernel} for $H_{\pm}^{2}$.
Clearly,%
\begin{equation}
\left\Vert k_{\lambda}\right\Vert =\sqrt{\langle k_{\lambda},k_{\lambda
}\rangle}=\frac{1}{\sqrt{2\operatorname{Im}\lambda}}\label{norm}%
\end{equation}
and hence $k_{\lambda}\in H_{\pm}^{2}$ if $\lambda\in\mathbb{C}^{\pm}$. The
main reason why reproducing kernels are convenient is the following%
\begin{subequations}
\begin{align}
f &  \in H^{2},\lambda\in\mathbb{C}^{+}\Longrightarrow f\left(  \lambda
\right)  =\langle f,k_{\lambda}\rangle\text{ (Cauchy's formula)}%
\label{form for rep ker}\\
f &  \in L^{2},\lambda\in\mathbb{R}\Longrightarrow(\mathbb{P}_{\pm}f)\left(
\lambda\right)  =\pm\langle f,k_{\lambda\pm i0}\rangle.\label{Riesz}%
\end{align}
Let $B$ be a \emph{Blaschke product} with finitely\footnote{It can also be
infinite but it doesn't concern us.} many simple zeros $z_{n}\in\mathbb{C}%
^{+}$, i.e.,%
\end{subequations}
\[
B\left(  z\right)  =\prod_{n}b_{n}\left(  z\right)  ,b_{n}\left(  z\right)
=\frac{z-z_{n}}{z-\overline{z_{n}}}.
\]
Introduce%
\[
K_{B}=\operatorname*{span}\left\{  k_{z_{n}}\right\}  .
\]
It is an easy but nevertheless fundamentally important fact in interpolation
of analytic functions, the study of the shift operator, so-called model
operators, etc. that%
\begin{equation}
K_{B}=H^{2}\ominus BH^{2},\text{ where }\ BH^{2}:=\left\{  Bf:f\in
H^{2}\right\}  .\label{K}%
\end{equation}

\begin{lemma}
\label{Blaschke lemma}The orthogonal projections $\mathbb{P}_{B}$ of $H^{2}$
onto $K_{B}$ and $\mathbb{P}_{B}^{\bot}=I-\mathbb{P}_{B}$ are given by%
\begin{equation}
\mathbb{P}_{B}=B\mathbb{P}_{-}\overline{B},\ \ \mathbb{P}_{B}^{\bot
}=B\mathbb{P}_{+}\overline{B}. \label{PB}%
\end{equation}
Furthermore, if $A$ is a linear bounded operator in $H^{2}$ then the matrix of
$\mathbb{P}_{B}A\mathbb{P}_{B}$ with respect to $\left(  k_{z_{n}}\right)  $
is given by%
\begin{equation}
\left(  \mathbb{P}_{B}AP_{B}\right)  _{mn}=\left\langle Ak_{z_{n}},k_{z_{m}%
}^{\bot}\right\rangle , \label{mn}%
\end{equation}
where%
\begin{equation}
k_{z_{n}}^{\bot}\left(  z\right)  :=\frac{2\operatorname{Im}z_{n}}%
{B_{n}\left(  z_{n}\right)  }B_{n}\left(  z\right)  k_{z_{n}}\left(  z\right)
,\ \ \ B_{n}:=B/b_{n} \label{bi orth}%
\end{equation}
form a bi-orthogonal basis for $\left(  k_{z_{n}}\right)  $. I.e.,
$\left\langle k_{z_{n}}^{\bot},k_{z_{m}}\right\rangle =\delta_{nm}$.
\end{lemma}

\begin{proof}
(\ref{PB}) are proven in \cite{NikBook86}. To show (\ref{mn}) we first
explicitly evaluate $\mathbb{P}_{B}$. By (\ref{Riesz proj}) for $f\in H^{2}$
we have%
\[
\mathbb{P}_{-}\overline{B}f=-\frac{1}{2\pi i}\int_{-\infty}^{\infty}%
\frac{f\left(  s\right)  }{B\left(  s\right)  }\frac{ds}{s-\left(
x-i0\right)  }%
\]
and by residues%
\begin{align*}
\left(  \mathbb{P}_{-}\overline{B}f\right)  \left(  x\right)   &  =-%
%TCIMACRO{\dsum _{n}}%
%BeginExpansion
{\displaystyle\sum_{n}}
%EndExpansion
\operatorname*{Res}\left(  \frac{f\left(  z\right)  /B\left(  z\right)  }%
{z-x},z_{n}\right)  \ \\
&  =%
%TCIMACRO{\dsum _{n}}%
%BeginExpansion
{\displaystyle\sum_{n}}
%EndExpansion
\frac{2i\operatorname{Im}z_{n}}{B_{n}\left(  z_{n}\right)  }\frac{f\left(
z_{n}\right)  }{x-z_{n}}\\
&  =%
%TCIMACRO{\dsum _{n}}%
%BeginExpansion
{\displaystyle\sum_{n}}
%EndExpansion
\frac{2i\operatorname{Im}z_{n}}{B_{n}\left(  z_{n}\right)  }\frac{\left\langle
f,k_{z_{n}}\right\rangle }{x-z_{n}}\text{ (by (\ref{form for rep ker})).}%
\end{align*}
Hence, by (\ref{PB}),%
\begin{align*}
\mathbb{P}_{B}f &  =%
%TCIMACRO{\dsum _{n}}%
%BeginExpansion
{\displaystyle\sum_{n}}
%EndExpansion
\ \left\langle f,k_{z_{n}}\right\rangle \ \frac{2i\operatorname{Im}z_{n}%
}{B_{n}\left(  z_{n}\right)  }B_{n}\ \frac{1}{\cdot-\overline{z}_{n}}\\
&  =%
%TCIMACRO{\dsum _{n}}%
%BeginExpansion
{\displaystyle\sum_{n}}
%EndExpansion
\ \left\langle f,k_{z_{n}}\right\rangle k_{z_{n}}^{\bot},
\end{align*}
where $k_{z_{n}}^{\bot}$ is given by (\ref{bi orth}). It remains to verifies
that $\left(  k_{z_{n}}^{\bot}\right)  $ forms a bi-orthogonal basis for
$K_{B}$. Indeed,%
\begin{align*}
\left\langle k_{z_{n}}^{\bot},k_{z_{m}}\right\rangle  &  =\left\langle
\frac{2\operatorname{Im}z_{n}}{B_{n}\left(  z_{n}\right)  }B_{n}k_{z_{n}%
},k_{z_{m}}\right\rangle =\frac{2\operatorname{Im}z_{n}}{B_{n}\left(
z_{n}\right)  }\left\langle B_{n}k_{z_{n}},k_{z_{m}}\right\rangle \\
&  =\frac{2\operatorname{Im}z_{n}}{B_{n}\left(  z_{n}\right)  }B_{n}\left(
z_{m}\right)  k_{z_{n}}\left(  z_{m}\right)  .
\end{align*}
If $n\neq m$ then $B_{n}\left(  z_{m}\right)  =0$. If $n=m$ then by
(\ref{norm})
\[
\left\langle k_{z_{n}}^{\bot},k_{z_{m}}\right\rangle =2\operatorname{Im}%
z_{n}\ k_{z_{n}}\left(  z_{n}\right)  =1.
\]
The formula (\ref{mn}) easily follows now.
\end{proof}

\subsection{Hankel operators\label{HO}}

A Hankel operator is an infinitely dimensional analog of a Hankel matrix, a
matrix whose $(j,k)$ entry depends only on $j+k$. In the context of integral
operators the Hankel operator is usually defined as an integral operator on
$L^{2}(\mathbb{R}_{+})$ whose kernel depends on the sum of the arguments%
\begin{equation}
(\mathbb{H}f)(x)=\int_{0}^{\infty}h(x+y)f(y)dy,\;f\in L^{2}(\mathbb{R}%
_{+}),\;x\geq0\label{eq4.10}%
\end{equation}
and it is this form that Hankel operators typically appear in the inverse
scattering formalism. It is much more convenient for our purposes to consider
\emph{Hankel operators} on $H^{2}$ (cf. \cite{Nik2002}, \cite{Peller2003}).

Let%
\[
(\mathbb{J}f)(x)=f(-x)
\]
be the operator of reflection on $L^{2}$ and let $\varphi\in L^{\infty}$. The
operators $\mathbb{H}(\varphi)$ defined by%
\begin{equation}
\mathbb{H}(\varphi)f=\mathbb{JP}_{-}\varphi f,\ \ \ f\in H^{2}, \label{eq4.1}%
\end{equation}
is called the Hankel operator with the symbol $\varphi$.

It is clear that $\mathbb{H}(\varphi)$ is bounded from $H^{2}$ to $H^{2}$ and
\begin{equation}
\mathbb{H}(\varphi+h)=\mathbb{H}(\varphi)\text{ for any }h\in H^{\infty}.
\label{Sarason}%
\end{equation}
It is also straightforward to verify that $\mathbb{H}(\varphi)$ is selfadjoint
if $\mathbb{J}\varphi=\bar{\varphi}.$

The following elementary lemma on Hankel operators with analytic symbols will
be particularly useful.

\begin{lemma}
\label{on split}Let a function $\varphi$ be meromorphic on $\mathbb{C}$ and
subject to%
\begin{equation}
\varphi\left(  -\overline{z}\right)  =\bar{\varphi}\left(  z\right)  \text{
(symmetry).} \label{sym of fi}%
\end{equation}
If $\varphi$ has finitely many simple poles $\left\{  z_{n}\right\}
_{n=-N}^{N}$ in $\mathbb{C}^{+}$, is bounded on $\mathbb{R}$, and for any
$h\geq0$
\begin{equation}
\varphi\left(  x+ih\right)  =O\left(  x^{-1}\right)  ,x\rightarrow\pm\infty,
\label{est}%
\end{equation}
then the Hankel operator $\mathbb{H}(\varphi)$ is selfadjoint, trace class,
and admits the decomposition%
\begin{equation}
\mathbb{H}(\varphi)=\mathbb{H}(\phi)+\mathbb{H}(\Phi), \label{basic split}%
\end{equation}
where $\phi$ is a rational function and $\Phi$ is an entire function given
respectively by%
\[
\phi\left(  x\right)  =%
%TCIMACRO{\dsum _{-N\leq n\leq N}}%
%BeginExpansion
{\displaystyle\sum_{-N\leq n\leq N}}
%EndExpansion
\frac{\operatorname*{Res}\left(  \varphi,z_{n}\right)  }{x-z_{n}},
\]%
\begin{equation}
\Phi\left(  x\right)  =-\frac{1}{2\pi i}\int_{\mathbb{R}+ih}\frac
{\varphi\left(  s\right)  }{s-x}\ ds,\ \ \ h>\max_{n}\operatorname{Im}z_{n}.
\label{big fi}%
\end{equation}
Moreover,%
\begin{equation}
\mathbb{H}(\phi)=%
%TCIMACRO{\dsum _{-N\leq n\leq N}}%
%BeginExpansion
{\displaystyle\sum_{-N\leq n\leq N}}
%EndExpansion
i\operatorname*{Res}\left(  \varphi,z_{-n}\right)  \left\langle \cdot
,k_{z_{-n}}\right\rangle k_{z_{n}}, \label{H(small fi)}%
\end{equation}%
\begin{equation}
\mathbb{H}(\Phi)=\int_{\mathbb{R}+ih}\frac{dz}{2\pi}\varphi\left(  z\right)
\left\langle \cdot,k_{z}\right\rangle k_{-\overline{z}}=\int_{\mathbb{R}%
+ih}\frac{dz}{2\pi}\varphi\left(  -\overline{z}\right)  \left\langle
\cdot,k_{-\overline{z}}\right\rangle k_{z}, \label{H(big fi)}%
\end{equation}
where $k_{\lambda}\left(  z\right)  =\frac{i}{z-\overline{\lambda}}$ is the
reproducing kernel of $H^{2}$.
\end{lemma}

\begin{proof}
The selfadjointness follows from (\ref{sym of fi}). By (\ref{Sarason})
\[
\mathbb{H}(\varphi)=\mathbb{H}(\mathbb{P}_{-}\varphi)
\]
and hence we have to worry only about $\mathbb{P}_{-}\varphi$. By by the
residue theorem ($h>\max_{k}\operatorname{Im}z_{k}$), we have%
\begin{align*}
\left(  \mathbb{P}_{-}\varphi\right)  \left(  x\right)   &  =-\frac{1}{2\pi
i}\int_{\mathbb{R}}\frac{\varphi\left(  s\right)  }{s-\left(  x-i0\right)
}\ ds\\
&  =%
%TCIMACRO{\dsum _{-N\leq n\leq N}}%
%BeginExpansion
{\displaystyle\sum_{-N\leq n\leq N}}
%EndExpansion
\frac{\operatorname*{Res}\left(  \varphi,z_{n}\right)  }{x-z_{n}}-\frac
{1}{2\pi i}\int_{\mathbb{R}+ih}\frac{\varphi\left(  s\right)  }{s-z}\ ds\\
&  =\phi\left(  x\right)  +\Phi\left(  x\right)  ,
\end{align*}
and (\ref{basic split}) follows. Apparently $\Phi$ is analytic (and bounded)
below the line $\mathbb{R}+ih$. Since $h$ is arbitrary, $\Phi$ is then entire.
Moreover, all derivatives of $\Phi$ are bounded on $\mathbb{R}$ and therefore
$\mathbb{H}(\Phi)$ is at least trace class (in any Shatten-von Neumann ideal).

It follows from (\ref{P_}) that for any $z\in\mathbb{C}^{+}$%
\[
\mathbb{H}(\frac{1}{\cdot-z})f=if\left(  z\right)  k_{-\overline{z}}%
\]
and (\ref{H(small fi)})-(\ref{H(big fi)}) follow.
\end{proof}

\begin{corollary}
\label{corollary on analytic fi}If $\varphi$ has no poles in $\mathbb{C}^{+}$
then $\mathbb{H}(\varphi)=\mathbb{H}(\Phi)$.
\end{corollary}

\begin{corollary}
\label{corollary on fi=0}If (\ref{est}) holds uniformly in $h\geq h_{0}%
>\max_{n}\operatorname{Im}z_{n}$ then $\Phi=0.$
\end{corollary}

A very important feature of analytic symbols is that $\mathbb{H}(\varphi)$ is
well-defined outside of $H^{2}$. In particular, $\mathbb{H}(\varphi)k_{x+i0}$
is a smooth element of $H^{2}$ for any $x\in\mathbb{R}$ while $k_{x+i0}%
\not \in H^{2}$. We will need the following statement.

\begin{corollary}
\label{smooth elements}For every $x,s\in\mathbb{R}$%
\begin{align}
\mathbb{H}(\Phi)k_{x}\left(  s\right)   &  =\lim_{\varepsilon\rightarrow
0}\mathbb{H}(\Phi)k_{x+i\varepsilon}\left(  s\right)  \nonumber\\
&  =-\int_{\mathbb{R}+ih}\frac{\Phi\left(  z\right)  }{\left(  z-x\right)
\left(  z+s\right)  }\frac{dz}{2\pi}\nonumber\\
&  =-\int_{\mathbb{R}+ih}\frac{\varphi\left(  z\right)  }{\left(  z-x\right)
\left(  z+s\right)  }\frac{dz}{2\pi}\label{K(s) 0}\\
&  =:K_{x}\left(  s\right)  \in C^{\infty}\left(  \mathbb{R}\right)  \cap
H^{2}.\nonumber
\end{align}
Moreover, if $\varphi_{\varepsilon}\rightarrow\varphi$ uniformly on
$\mathbb{R}+ih$ then for every $x,s\in\mathbb{R}$%
\begin{equation}
\lim_{\varepsilon\rightarrow0}\mathbb{H}(\Phi_{\varepsilon})k_{x+i\varepsilon
}\left(  s+i\varepsilon\right)  =K_{x}\left(  s\right)  .\label{K(s) 1}%
\end{equation}
Convergence in (\ref{K(s) 0}) and (\ref{K(s) 1}) also holds in $L^{2}$.
\end{corollary}

\begin{proof}
It follows from (\ref{H(big fi)}) that%
\begin{align*}
\mathbb{H}(\Phi)k_{x+i\varepsilon}\left(  s\right)   &  =\int_{\mathbb{R}%
+ih}\frac{dz}{2\pi}\varphi\left(  z\right)  \left\langle k_{x+i\varepsilon
},k_{z}\right\rangle k_{-\overline{z}}\left(  s\right)  \\
&  =\int_{\mathbb{R}+ih}\frac{dz}{2\pi}\varphi\left(  z\right)  \left\langle
k_{x+i\varepsilon},k_{z}\right\rangle k_{-\overline{z}}\left(  s\right)
\text{ \ \ \ \ \ \ \ \ \ \ (by (\ref{form for rep ker}))}\\
&  =-\int_{\mathbb{R}+ih}\frac{dz}{2\pi}\varphi\left(  z\right)
k_{x+i\varepsilon}\left(  z\right)  k_{-\overline{z}}\left(  s\right)  \\
&  \rightarrow-\int_{\mathbb{R}+ih}\frac{\varphi\left(  z\right)  }{\left(
z-x\right)  \left(  z+s\right)  }\frac{dz}{2\pi}=-\int_{\mathbb{R}+ih}%
\frac{\Phi\left(  z\right)  }{\left(  z-x\right)  \left(  z+s\right)  }%
\frac{dz}{2\pi},\ \\
\varepsilon &  \rightarrow0,
\end{align*}
where we have used two obvious facts: (a) $k_{x+i\varepsilon}\left(  z\right)
\rightarrow k_{x}\left(  z\right)  $ uniformly on $\mathbb{R}+ih$, and (b) by
the Lebesgue dominated convergence%
\[
\int_{\mathbb{R}+ih}\frac{\phi\left(  z\right)  }{\left(  z-x\right)  \left(
z+s\right)  }\frac{dz}{2\pi}=\lim_{h\rightarrow\infty}\int_{\mathbb{R}%
+ih}\frac{\phi\left(  z\right)  }{\left(  z-x\right)  \left(  z+s\right)
}\frac{dz}{2\pi}=0.
\]
Thus (\ref{K(s) 0}) is proven. (\ref{K(s) 1}) is proven similarly.
\end{proof}

\section{Our explicit potential and its short-range
approximation\label{our explicit Q}}

In this section we explicitly construct a symmetric Wigner-von Neumann type
potentials supporting one negative and one positive bound state. Our
construction is base upon a classical Gelfand-Levitan example
\cite{LevitanInverseProblems87} of an explicit potential of a half-line
Schr\"{o}dinger operator which spectral measure has one positive pure point.
The symmetric extension of this potential to the whole line will be our
initial condition. We then find its explicit short range approximation, which
will be crucial to our consideration.

\subsection{An explicit WvN type potential\label{subs1}}

Consider the function%
\begin{equation}
m\left(  \lambda\right)  =i\sqrt{\lambda}+\frac{2\rho}{1-\lambda
},\ \ \ \operatorname{Im}\lambda\geq0,\label{m}%
\end{equation}
where $\rho$ is some positive number. This is a Herglotz function (i.e.
analytic function mapping $\mathbb{C}^{+}$\ to $\mathbb{C}^{+}$) which
coincides with the Titchmarsh-Weyl $m-$function\footnote{We recall that the
problem $-\partial_{x}^{2}u+q(x)u=\lambda u,\ x\in\left(  0,\pm\infty\right)
,\ u\left(  \pm0,\lambda\right)  =1$ has a unique square integrable
(\emph{Weyl}) solution $\Psi_{\pm}(x,\lambda)$ for any $\operatorname{Im}%
\lambda>0$ for broad classes of $q$'s (called\emph{\ limit point case}).
Define then the \emph{(Titchmarsh-Weyl) m-function} $m_{\pm}$ for $\left(
0,\pm\infty\right)  $ as follows: $m_{\pm}\left(  \lambda\right)  =\pm
\partial_{x}\Psi_{\pm}\left(  \pm0,\lambda\right)  $.} of the (Dirichlet)
Schr\"{o}dinger operators $-d^{2}/dx^{2}+q_{0}\left(  x\right)  $ on
$L^{2}\left(  0,\infty\right)  $ with a Dirichlet boundary condition at $0$.
The potential $q_{0}$ has the following explicit form%
\begin{equation}
q_{0}\left(  x\right)  =-2\frac{d^{2}}{dx^{2}}\log\tau_{0}\left(  x\right)
,\ \ \ x\geq0,\label{Q}%
\end{equation}
where%
\begin{equation}
\tau_{0}\left(  x\right)  =1+2\rho\int_{0}^{x}\sin^{2}s\ ds=1+\rho x-\left(
\rho/2\right)  \sin2x.\label{tua for t=0}%
\end{equation}
Introduce%
\begin{equation}
Q\left(  x\right)  =\left\{
\begin{array}
[c]{cc}%
q_{0}\left(  x\right)  , & x\geq0\\
q_{0}\left(  -x\right)  , & x<0
\end{array}
\right.  ,\label{our Q}%
\end{equation}
i.e., $Q$ is an even extension of $q_{0}$. One can easily see that the
function $Q$ is continuous and $Q\left(  0\right)  =0$ but not continuously
differentiable. In fact, $Q\,$\ is as smooth at $x=0$ as $\left\vert \sin
x\right\vert $. Moreover, one has
\begin{equation}
Q\left(  x\right)  =-4\ \dfrac{\sin2x}{x}+O\left(  \frac{1}{x^{2}}\right)
,\ \ x\rightarrow\pm\infty,\label{Q asym}%
\end{equation}
and hence $Q\in L^{2}\left(  \mathbb{R}\right)  $ but $\left(  1+\left\vert
x\right\vert \right)  Q\left(  x\right)  $ is not in $L^{1}\left(
\mathbb{R}\right)  $. Thus, $Q$ is not short-range. Also note that%
\[
\int_{-\infty}^{\infty}Q\left(  x\right)  dx=0.
\]
The main feature of $Q$ is that $\mathbb{L}_{Q}$ admits an explicit spectral
and scattering theory.

\begin{theorem}
\label{Thm on Q}The Schr\"{o}dinger operator $\mathbb{L}_{Q}$ on $L^{2}\left(
\mathbb{R}\right)  $ with $Q$ given by (\ref{our Q}) has the following properties:

\begin{enumerate}
\item (Spectrum) The spectrum of $\mathbb{L}_{Q}$ consists of the two fold
absolutely continuos part filling $\left(  0,\infty\right)  $, one negative
bound state $-\kappa^{2}$ found from the real solution of%
\begin{equation}
\kappa^{3}+\kappa=2\rho\label{neg BS}%
\end{equation}
and one positive (embedded) bound state $+1$.

\item (Scattering quantities) For the norming constant $c$ of $-\kappa^{2}$ we
have%
\begin{equation}
c=-i\operatorname*{Res}\left(  T\left(  k\right)  ,i\kappa\right)
=-i\operatorname*{Res}\left(  R\left(  k\right)  ,i\kappa\right)  =\frac
{2\rho}{3\kappa^{2}+1} \label{c}%
\end{equation}
and for the scattering matrix we have%
\begin{equation}
S\left(  k\right)  =\left(
\begin{array}
[c]{cc}%
T\left(  k\right)  & R\left(  k\right) \\
R\left(  k\right)  & T\left(  k\right)
\end{array}
\right)  ,\ \ \ k\in\mathbb{R}, \label{S matrix}%
\end{equation}
where $T$ and $R$ are, respectively, the transmission and reflection
coefficients given by%
\begin{equation}
T\left(  k\right)  =\frac{P\left(  k\right)  }{P\left(  k\right)  +2i\rho
},\ \ \ R\left(  k\right)  =\frac{-2i\rho}{P\left(  k\right)  +2i\rho},
\label{TLR}%
\end{equation}%
\[
P\left(  k\right)  :=k^{3}-k.
\]

\end{enumerate}
\end{theorem}

\begin{proof}
Due to symmetry $m_{-}=m_{+}=m$ it follows from the general theory
\cite{Titchmarsh62} that the eigenvalues of the Schr\"{o}dinger operator
$\mathbb{L}_{Q}$ are the (necessarily simple) poles of $m$ and $1/m$. Thus,
$\mathbb{L}_{Q}$ has one positive bound state $+1$ (the pole of $m\left(
\lambda\right)  $) and one negative bound state $-\kappa^{2}$ (the zero of
$m\left(  \lambda\right)  $). Clearly (\ref{neg BS}) holds. The fact about the
absolutely continuos spectrum also follows from the general theory (as well as
from (2) below) and therefore (1) is proven.

Turn to (2). By a direct computation one verifies that%
\begin{align*}
f_{\pm}\left(  x,k\right)   &  =\left\{  1\pm\left(  \frac{e^{\pm ix}}%
{k+1}-\frac{e^{\mp ix}}{k-1}\right)  \frac{\rho\sin x}{1+\rho\left\vert
x\right\vert -\left(  \rho/2\right)  \sin2\left\vert x\right\vert }\right\}
e^{\pm ikx},\ \ \ \\
\pm x  &  \geq0,
\end{align*}
solve the Schr\"{o}dinger equation $\mathbb{L}_{Q}f=k^{2}f$ for $\pm x\geq0$
if $k\neq\pm1$. Since clearly%
\[
f_{\pm}\left(  x,k\right)  =\left(  1+o\left(  1\right)  \right)  e^{\pm
ikx},\ \ \ x\rightarrow\pm\infty,
\]
we can claim that $f_{\pm}$ are Jost solution corresponding to $\pm\infty$. By
the general formulas (see e.g. \cite{HintonKlausIP1989})%
\begin{subequations}
\begin{equation}
T\left(  k\right)  =\frac{1}{f_{-}\left(  k\right)  f_{+}\left(  k\right)
}\ \frac{2ik}{m_{+}\left(  k^{2}\right)  +m_{-}\left(  k^{2}\right)  }\text{
\ \ (transmission coefficient),} \label{T}%
\end{equation}%
\end{subequations}
\begin{subequations}
\begin{equation}
R\left(  k\right)  =-\frac{\overline{f_{+}\left(  k\right)  }}{f_{+}\left(
k\right)  }\ \frac{\overline{m_{+}\left(  k^{2}\right)  }+m_{-}\left(
k^{2}\right)  }{m_{+}\left(  k^{2}\right)  +m_{-}\left(  k^{2}\right)  }\text{
\ \ (right reflection coefficient),} \label{R}%
\end{equation}%
\end{subequations}
\begin{subequations}
\begin{equation}
L\left(  k\right)  =-\frac{\overline{f_{-}\left(  k\right)  }}{f_{-}\left(
k\right)  }\ \frac{m_{+}\left(  k^{2}\right)  +\overline{m_{-}\left(
k^{2}\right)  }}{m_{+}\left(  k^{2}\right)  +m_{-}\left(  k^{2}\right)
}\text{ \ \ (left reflection coefficient)} \label{L}%
\end{equation}
and $f_{\pm}\left(  k\right)  :=f_{\pm}\left(  0,k\right)  $ are Jost
functions. Since in our case $m_{\pm}=m$ and $f_{\pm}\left(  k\right)  =1$, we
immediately see that $L=R$ and arrive at (\ref{S matrix}).

It remains to demonstrate (\ref{c}). Recall the general fact (see e.g.
\cite{AK01}) that for any short-range $q$%
\end{subequations}
\begin{equation}
\operatorname*{Res}\left(  T,i\kappa_{n}\right)  =i\left(  -1\right)
^{n-1}\sqrt{c_{n}^{+}c_{n}^{-}},\label{norm const general statement}%
\end{equation}
where $c_{n}^{\pm}$ are right/left norming constant associated with the bound
states $-\kappa_{n}^{2}$ ($n=1,2,...)$ enumerated in the increasing order. If
$q$ is even then $c_{n}^{+}=c_{n}^{-}=c_{n}$ and hence in our case of a single
bound state $-\kappa^{2}$ we have%
\[
\operatorname*{Res}\left(  T,i\kappa\right)  =ic
\]
and the first equation in (\ref{c}) follows. The second and third equations in
(\ref{c}) can be verified by a direct computation.
\end{proof}

\begin{remark}
Same way as we did in the proof, one can find an analog of Theorem
\ref{Thm on Q} for the truncated potentials $Q|\mathbb{R}_{\pm}$. There will
be no positive bound state but the formulas (\ref{T})-(\ref{L}) immediately
yield same (\ref{TLR}) where $2\rho$ is replaced with $\rho$. Indeed, for
$Q|\mathbb{R}_{+}$ we have%
\[
m_{+}\left(  k^{2}\right)  =m\left(  k^{2}\right)  =ik+\frac{2\rho}{1-k^{2}%
},\ \ \ m_{-}\left(  k^{2}\right)  =ik,\ \ \ f_{\pm}\left(  k\right)  =1,
\]
and the claim follows. Moreover, (\ref{c}) also holds for the truncated $Q$
with the same substitution. This demonstrates clearly that the standard triple
$\left(  R,\kappa,c\right)  $ no longer constitutes scattering data.
\end{remark}

\subsection{Short-range approximation of $Q$}

The simples short range approximation is based upon a truncation but the
limiting procedure will not be simple. We instead approximate the scattering
data. While much more complicated than truncation, the limiting procedure
becomes easier to track.

If you recall the famous characterization of the scattering matrix
\cite{MarchBook2011} of a short-range potential, one of the conditions is that
$T\left(  k\right)  $ can vanish on $\overline{\mathbb{C}^{+}}$ only at $k=0$.
But in our case this occurs if $P\left(  k\right)  =0$ which happens also for
$k=\pm1$. This prompts to replace $P\left(  k\right)  $ in $T\left(  k\right)
$ given by (\ref{TLR}) with $P\left(  k\right)  +i\varepsilon$ with some small
$\varepsilon>0$. Clearly%
\[
P\left(  k\right)  +i\varepsilon=k^{3}-k+i\varepsilon=\left(  k-\mu
_{\varepsilon}\right)  \left(  k+\overline{\mu_{\varepsilon}}\right)  \left(
k-i\nu_{\varepsilon}\right)  ,
\]
where%
\begin{equation}
\mu_{\varepsilon}=1-i\varepsilon/2+O\left(  \varepsilon^{2}\right)
,\ \ \nu_{\varepsilon}=\varepsilon+O\left(  \varepsilon^{2}\right)
,\ \ \ \varepsilon\rightarrow0.\label{zeros below}%
\end{equation}
Thus two real zeros $\pm1$ move to $\mathbb{C}^{-}$. Form the Blaschke product
$B^{\varepsilon}$ with zeros $z_{-1}=-\mu,z_{1}=\overline{\mu},z_{0}=i\nu
\in\mathbb{C}^{+}$. I.e.,%
\[
B^{\varepsilon}=b_{-1}b_{0}b_{1},\ b_{n}\left(  k\right)  =\frac{k-z_{n}%
}{k-\overline{z_{n}}}\ .\
\]
It follows from (\ref{zeros below}) that as $\varepsilon\rightarrow0$%
\[
z_{n}=n+i\varepsilon/2^{\left\vert n\right\vert }+O\left(  \varepsilon
^{2}\right)  ,\ \ n=0,\pm1.
\]
The Blaschke product $B^{\varepsilon}$ will be a building block in our
approximation. Apparently, $B^{\varepsilon}\rightarrow1$ as $\varepsilon
\rightarrow0$ uniformly on compacts in $\mathbb{C}^{+}$ and a.e. on
$\mathbb{R}$. We are now ready to present our approximation.

\begin{theorem}
\label{Thm on approx}Let ($\varepsilon>0$)%
\begin{align}
T_{\varepsilon}\left(  k\right)   &  =\frac{\left(  P\left(  k\right)
+i\varepsilon\right)  ^{2}/b_{0}^{2}\left(  k\right)  }{P\left(  k\right)
+i\rho\left(  1+a\right)  }\frac{1}{P\left(  k\right)  +i\rho\left(
1-a\right)  },\ \ \ \nonumber\\
R_{\varepsilon}\left(  k\right)   &  =\frac{-2ia\rho}{P\left(  k\right)
+i\rho\left(  1+a\right)  }\frac{P\left(  k\right)  }{P\left(  k\right)
+i\rho\left(  1-a\right)  }\frac{1}{B^{\varepsilon}\left(  k\right)
},\label{R_epsilon}\\
a  &  :=\sqrt{1-\left(  \varepsilon/\rho\right)  ^{2}}.\nonumber
\end{align}
Then

\begin{enumerate}
\item The matrix\bigskip%
\[
S_{\varepsilon}=\left(
\begin{array}
[c]{cc}%
T_{\varepsilon} & R_{\varepsilon}\\
R_{\varepsilon} & T_{\varepsilon}%
\end{array}
\right)
\]
is the scattering matrix of a short-range potential having two bound states
$-\left(  \kappa_{\pm}^{\varepsilon}\right)  ^{2},\ \kappa_{+}^{\varepsilon
}>\kappa_{-}^{\varepsilon}$, subject to%
\begin{equation}
\kappa_{+}^{\varepsilon}=\kappa+O\left(  \varepsilon^{2}\right)
,\ \ \ \kappa_{-}^{\varepsilon}=O\left(  \varepsilon^{2}\right)
,\ \ \varepsilon\rightarrow0. \label{neg BS of approx}%
\end{equation}

\item If we choose the left and right norming constants associated with
$-\left(  \kappa_{\pm}^{\varepsilon}\right)  ^{2}$ equal to each other and to
satisfy%
\begin{equation}
c_{\pm}^{\varepsilon}=\mp i\operatorname*{Res}\left(  T_{\varepsilon}%
,i\kappa_{\pm}^{\varepsilon}\right)  , \label{norm constants for approx}%
\end{equation}
then the unique potential $Q_{\varepsilon}\left(  x\right)  $ corresponding to
the scattering data%
\[
\left\{  R_{\varepsilon},\kappa_{\pm}^{\varepsilon},c_{\pm}^{\varepsilon
}\right\}
\]
is even and everywhere%
\begin{equation}
Q_{\varepsilon}\left(  x\right)  \rightarrow Q\left(  x\right)
,\ \ \ \varepsilon\rightarrow0. \label{lim pot}%
\end{equation}

\end{enumerate}
\end{theorem}

\begin{proof}
To prove part 1 of the statement one needs to check all the conditions of the
Marchenko characterization \cite{MarchBook2011}. Is is straightforward but
quite involved and we omit it. By the general theory, the bound states are the
squares of the (simple) poles of $T_{\varepsilon}$ in $\mathbb{C}^{+}$, i.e.
the solutions of two%
\[
P\left(  k\right)  +i\rho\left(  1\pm a\right)  =0,\ B^{\varepsilon}\left(
k\right)  =0.
\]
Since each equation has only one imaginary solution $i\kappa_{\pm
}^{\varepsilon}$, we have exactly two bound states $-\left(  \kappa_{\pm
}^{\varepsilon}\right)  ^{2}$ which are clearly subject to
(\ref{neg BS of approx}). Note that $z_{0}$, the zero of $b_{0}$, is a
removable singularity by our very construction of $B^{\varepsilon}$.

Turn now to part 2. Consider the reflection coefficient $R_{\varepsilon}$.
Apparently, $R_{\varepsilon}$ is a rational function with five simple poles.
Two imaginary poles $i\kappa_{\pm}^{\varepsilon}$ are shared with
$T_{\varepsilon}$ plus $z_{n},n=0,\pm1$, the zeros of $B^{\varepsilon}\left(
k\right)  $. By a direct computation, one checks%
\begin{align}
\operatorname*{Res}\left(  R_{\varepsilon},i\kappa_{\pm}^{\varepsilon}\right)
&  =\pm\operatorname*{Res}\left(  T_{\varepsilon},i\kappa_{\pm}^{\varepsilon
}\right)  \nonumber\\
&  =ic_{\pm}^{\varepsilon}\text{ (since (\ref{norm constants for approx}%
)).}\label{+-}%
\end{align}
We now solve the inverse scattering problem for the data
\[
\left\{  R_{\varepsilon},\kappa_{\pm}^{\varepsilon},-i\operatorname*{Res}%
\left(  R_{\varepsilon},i\kappa_{\pm}^{\varepsilon}\right)  \right\}  ,
\]
basing upon our Hankel operator approach \cite{GruRybSIMA15}. To this end,
form the symbol%
\begin{equation}
\varphi_{x}^{\varepsilon}\left(  k\right)  =\frac{-\operatorname*{Res}\left(
R_{\varepsilon},i\kappa_{+}^{\varepsilon}\right)  }{k-i\kappa_{+}%
^{\varepsilon}}e^{-2\kappa_{+}^{\varepsilon}x}+\frac{-\operatorname*{Res}%
\left(  R_{\varepsilon},i\kappa_{-}^{\varepsilon}\right)  }{k-i\kappa
_{-}^{\varepsilon}}e^{-2\kappa_{-}^{\varepsilon}x}+R_{\varepsilon}\left(
k\right)  e^{2ikx}.\label{fi^epsilon}%
\end{equation}
One immediately sees that $\varphi_{x}^{\varepsilon}$ is subject to the
conditions of Lemma \ref{on split} with three (symmetric) poles $z_{n}%
,n=0,\pm1$. By condition, the left and right scattering data are identical and
hence $Q_{\varepsilon}$ must be even and it enough to recover it only on
$\left(  0,\infty\right)  $. Therefore we can assume that $x>0$ in
(\ref{fi^epsilon}) which by Corollary \ref{corollary on fi=0} implies that the
$\Phi$-part of our symbol is zero. By Lemma \ref{on split}
\[
\mathbb{H}(\varphi_{x}^{\varepsilon})=%
%TCIMACRO{\dsum _{-1\leq n\leq1}}%
%BeginExpansion
{\displaystyle\sum_{-1\leq n\leq1}}
%EndExpansion
i\operatorname*{Res}\left(  \varphi_{x}^{\varepsilon},z_{-n}\right)
\left\langle \cdot,k_{z_{-n}}\right\rangle k_{z_{n}}.
\]
Thus, our Hankel operator is rank 3 and by the Dyson formula
\cite{GruRybSIMA15} we have%
\[
Q_{\varepsilon}\left(  x\right)  =-2\partial_{x}^{2}\log\det\left(
I+\mathbb{H}(\varphi_{x}^{\varepsilon})\right)  ,x>0.
\]
Note that our $Q_{\varepsilon}$ has an exponential decay and can be explicitly
evaluated. We however don't really need it. We will take the limit as
$\varepsilon\rightarrow0$ in the next section.
\end{proof}

We emphasize that Part 2 of Theorem \ref{Thm on approx} is essential because,
due to nonuniqueness, it is a priori unclear if our approximations indeed
converges to the original potential.

Note also, that $Q_{\varepsilon}\left(  x\right)  $ all have the property that
$T_{\varepsilon}\left(  0\right)  \neq0$. Such potentials are called
exceptional because generically $T\left(  0\right)  =0.$

\section{Main results\label{main results}}

Through this section
\[
\xi_{x,t}(k)=\exp\{i(8k^{3}t+2kx)\}.
\]

\begin{theorem}
\label{main thm}Let $Q$ be the initial condition (\ref{our Q}) in the KdV
equation (\ref{KdV}),%
\[
\varphi_{x,t}\left(  k\right)  =R\left(  k\right)  \xi_{x,t}(k)-\frac
{\operatorname*{Res}\left(  R\xi_{x,t},i\kappa\right)  }{k-i\kappa},
\]
and $\mathbb{H}_{x,t}:=\mathbb{H}\left(  \varphi_{x,t}\right)  $, the
associated Hankel operator. Then (\ref{KdV}) has the (unique) classical
solution given by%
\begin{equation}
u\left(  x,t\right)  =u_{0}\left(  x,t\right)  +u_{1}\left(  x,t\right)
\label{u}%
\end{equation}
where%
\begin{equation}
u_{0}\left(  x,t\right)  =-2\partial_{x}^{2}\log\det\left\{  I+\mathbb{H}%
_{x,t}\right\}  ,\label{u0}%
\end{equation}
and%
\[
u_{1}\left(  x,t\right)  =-2\partial_{x}^{2}\log\tau\left(  x,t\right)  ,
\]%
\begin{align*}
\tau\left(  x,t\right)   &  =1+\rho\left(  x+12t\right)  -\frac{\rho}{2}%
\sin\left(  2x+8t\right)  \\
&  +\frac{\rho}{2}\operatorname{Re}\left.  \left(  I+\mathbb{H}_{x,t}\right)
^{-1}\left(  \mathbb{H}_{x,t}k_{1+i0}-\xi_{x,t}\left(  1\right)
\mathbb{H}_{x,t}k_{-1+i0}\right)  \right\vert _{1+i0}.
\end{align*}
Here, as before,$\ k_{\lambda}\left(  s\right)  =\dfrac{i}{s-\overline
{\lambda}}$ is the reproducing kernel.
\end{theorem}

\begin{proof}
Since our approximation $Q_{\varepsilon}\left(  x\right)  $ decays
exponentially, the (classical) solution to the KdV equation can be found in
closed form by Dyson's formula%
\begin{equation}
Q_{\varepsilon}\left(  x,t\right)  =-2\partial_{x}^{2}\log\det\left(
I+\mathbb{H}\left(  \varphi_{x,t}^{\varepsilon}\right)  \right)
,\label{q sub epsilon}%
\end{equation}
where%
\begin{align}
\varphi_{x,t}^{\varepsilon}\left(  k\right)   &  =\frac{-\operatorname*{Res}%
\left(  R_{\varepsilon},i\kappa_{+}^{\varepsilon}\right)  }{k-i\kappa
_{+}^{\varepsilon}}\xi_{x,t}\left(  i\kappa_{+}^{\varepsilon}\right)
+\frac{-\operatorname*{Res}\left(  R_{\varepsilon},i\kappa_{-}^{\varepsilon
}\right)  }{k-i\kappa_{-}^{\varepsilon}}\xi_{x,t}\left(  i\kappa
_{-}^{\varepsilon}\right)  \label{fi epsilon}\\
&  +R_{\varepsilon}\left(  k\right)  \xi_{x,t}\left(  k\right)  .\nonumber
\end{align}
Note that due to the Bourgain theorem \cite{Bourgain93}\ the limit
$\lim_{\varepsilon\rightarrow0}Q_{\varepsilon}\left(  x,t\right)  $ does exist
but we cannot pass to the limit in (\ref{q sub epsilon}) under the determinant
sign since, as we will see later, $\mathbb{H}\left(  \varphi_{x,t}%
^{\varepsilon}\right)  $ doesn't converge in the trace norm to $\mathbb{H}%
\left(  \varphi_{x,t}\right)  $, where%
\[
\varphi_{x,t}\left(  k\right)  =\frac{-\operatorname*{Res}\left(
R,i\kappa\right)  }{k-i\kappa}\xi_{x,t}\left(  i\kappa\right)  +R\left(
k\right)  \xi_{x,t}\left(  k\right)  .
\]
To detour this circumstance we split our determinant as follows. Consider%
\[
K_{\varepsilon}=\operatorname*{span}\left\{  k_{z_{n}}\right\}  _{n=-1}^{1},
\]
and decompose $H^{2}$ into the orthogonal sum (see Subsection
\ref{rep kernels})%
\begin{equation}
H^{2}=K_{\varepsilon}\oplus K_{\varepsilon}^{\bot},\ \ \ K_{\varepsilon}%
^{\bot}=B^{\varepsilon}H^{2}.\label{orthog decomp}%
\end{equation}
The decomposition (\ref{orthog decomp}) induces the block representation%
\[
\mathbb{H}\left(  \varphi_{x,t}^{\varepsilon}\right)  =\left(
\begin{array}
[c]{cc}%
\mathbb{H}_{0} & \mathbb{H}_{01}\\
\mathbb{H}_{01}^{\ast} & \mathbb{H}_{1}%
\end{array}
\right)  ,
\]
where%
\begin{align*}
\mathbb{H}_{0} &  :=\mathbb{P}_{B^{\varepsilon}}\mathbb{H}\left(
\varphi_{x,t}^{\varepsilon}\right)  \mathbb{P}_{B^{\varepsilon}}%
,\ \ \mathbb{H}_{1}=\mathbb{P}_{B^{\varepsilon}}^{\bot}\mathbb{H}\left(
\varphi_{x,t}^{\varepsilon}\right)  \mathbb{P}_{B^{\varepsilon}}^{\bot}\\
\mathbb{H}_{01} &  :=\mathbb{P}_{B^{\varepsilon}}^{\bot}\mathbb{H}\left(
\varphi_{x,t}^{\varepsilon}\right)  \mathbb{P}_{B^{\varepsilon}}%
,\ \ \mathbb{H}_{01}^{\ast}=\mathbb{P}_{B^{\varepsilon}}\mathbb{H}\left(
\varphi_{x,t}^{\varepsilon}\right)  \mathbb{P}_{B^{\varepsilon}}^{\bot},
\end{align*}
and%
\begin{align*}
\varphi_{x,t}^{\varepsilon}\left(  k\right)   &  =\frac{-\operatorname*{Res}%
\left(  R_{\varepsilon},i\kappa_{+}^{\varepsilon}\right)  }{k-i\kappa
_{+}^{\varepsilon}}\xi_{x,t}\left(  i\kappa_{+}^{\varepsilon}\right)
+\frac{-\operatorname*{Res}\left(  R_{\varepsilon},i\kappa_{-}^{\varepsilon
}\right)  }{k-i\kappa_{-}^{\varepsilon}}\xi_{x,t}\left(  i\kappa
_{-}^{\varepsilon}\right)  \\
&  +R_{\varepsilon}\left(  k\right)  \xi_{x,t}\left(  k\right)  .
\end{align*}
Examine the block $\mathbb{H}_{1}$ first. It follows from (\ref{fi epsilon})
that the poles of $\varphi_{x,t}^{\varepsilon}$ coincide with zeros $\left(
z_{n}\right)  $ of $B^{\varepsilon}$ and therefore by Lemma \ref{on split}
($h>\kappa_{+}^{\varepsilon}$)%
\begin{align*}
\mathbb{H}\left(  \varphi_{x,t}^{\varepsilon}\right)   &  =%
%TCIMACRO{\dsum _{-1\leq n\leq1}}%
%BeginExpansion
{\displaystyle\sum_{-1\leq n\leq1}}
%EndExpansion
i\operatorname*{Res}\left(  \xi_{x,t}R_{\varepsilon},z_{-n}\right)
\left\langle \cdot,k_{z_{-n}}\right\rangle k_{z_{n}}\\
&  +\int_{\mathbb{R}+ih}\frac{dz}{2\pi}\varphi_{x,t}^{\varepsilon}\left(
z\right)  \left\langle \cdot,k_{z}\right\rangle k_{-\overline{z}}\\
&  =%
%TCIMACRO{\dsum _{-1\leq n\leq1}}%
%BeginExpansion
{\displaystyle\sum_{-1\leq n\leq1}}
%EndExpansion
i\operatorname*{Res}\left(  \xi_{x,t}R_{\varepsilon},z_{-n}\right)
\left\langle \cdot,k_{z_{-n}}\right\rangle k_{z_{n}}+\mathbb{H}\left(
\Phi_{x,t}^{\varepsilon}\right)  .
\end{align*}
One immediately sees that%
\[
\mathbb{H}_{1}=\mathbb{P}_{B^{\varepsilon}}^{\bot}\mathbb{H}\left(  \Phi
_{x,t}^{\varepsilon}\right)  \mathbb{P}_{B^{\varepsilon}}^{\bot}.
\]
Since $\varphi_{x,t}^{\varepsilon}\rightarrow\varphi_{x,t}$ uniformly on
$\mathbb{R}+ih$, we obviously have
\begin{align*}
\Phi_{x,t}^{\varepsilon}\left(  x\right)   &  =-\frac{1}{2\pi i}%
\int_{\mathbb{R}+ih}\frac{\varphi_{x,t}^{\varepsilon}\left(  s\right)  }%
{s-x}\ ds\\
&  \rightarrow-\frac{1}{2\pi i}\int_{\mathbb{R}+ih}\frac{\varphi_{x,t}\left(
s\right)  }{s-x}\ ds=\Phi_{x,t}\left(  x\right)  ,\ \ \ \varepsilon
\rightarrow0,
\end{align*}
in $C^{n}\left(  \mathbb{R}\right)  $ for any $n$ which in turn implies
\cite{Peller2003} that $\lim_{\varepsilon\rightarrow0}\mathbb{H}\left(
\Phi_{x,t}^{\varepsilon}\right)  =\mathbb{H}\left(  \Phi_{x,t}\right)  $ in
the trace norm (in fact in all $\mathfrak{S}_{p},p>0$). Since%
\[
\varphi_{x,t}\left(  k\right)  =\frac{-\operatorname*{Res}\left(
R,i\kappa\right)  }{k-i\kappa}\xi_{x,t}\left(  i\kappa\right)  +R\left(
k\right)  \xi_{x,t}\left(  k\right)  ,
\]
we see that $i\kappa$ is a removable singularity for $\varphi_{x,t}$ and hence
by Corollary \ref{corollary on analytic fi}%
\[
\mathbb{H}\left(  \Phi_{x,t}\right)  =\mathbb{H}\left(  \varphi_{x,t}\right)
.
\]
Since $B^{\varepsilon}\rightarrow1$ a.e., it follows from (\ref{PB}) that in
the strong operator topology%
\begin{equation}
\mathbb{P}_{B^{\varepsilon}}^{\bot}=B^{\varepsilon}\mathbb{P}_{+}%
\overline{B^{\varepsilon}}\rightarrow I,\ \ \varepsilon\rightarrow
0.\label{conv of proj}%
\end{equation}
But \cite{BotSil02}, if $H_{n}\rightarrow H$ in trace norm, $A_{n}$ is
self-adjoint, $\sup_{n}\left\Vert A_{n}\right\Vert <\infty$, and
$A_{n}\rightarrow A$ strongly, then $A_{n}H_{n}A_{n}\rightarrow AHA$ in trace
norm. Therefore, we can conclude that in trace norm%
\begin{equation}
\mathbb{H}_{1}\rightarrow\mathbb{H}\left(  \varphi_{x,t}\right)
,\ \ \varepsilon\rightarrow0.\label{conv in trace}%
\end{equation}
We now make use of a well-known formula from matrix theory:%
\begin{equation}
\det\left(
\begin{array}
[c]{cc}%
A_{11} & A_{12}\\
A_{21} & A_{22}%
\end{array}
\right)  =\det A_{11}\det\left(  A_{22}-A_{21}A_{11}^{-1}A_{12}\right)
,\label{block mat}%
\end{equation}
which yields%
\begin{align}
&  \det\left(  I+\mathbb{H}\left(  \varphi_{x,t}^{\varepsilon}\right)
\right)  \label{det decom 0}\\
&  =\det\left\{  I+\mathbb{H}_{1}\right\}  \cdot\det\left\{  I+\mathbb{H}%
_{0}-\mathbb{H}_{01}^{\ast}\left(  I+\mathbb{H}_{1}\right)  ^{-1}%
\mathbb{H}_{01}\right\}  .\nonumber
\end{align}
Our goal is to study what happens to (\ref{det decom 0}) as $\varepsilon
\rightarrow0$. The determinants on the right hand side of (\ref{det decom 0})
behave very differently and we treat them separately. It follows from
(\ref{conv in trace}) that%
\begin{equation}
\lim_{\varepsilon\rightarrow0}\det\left\{  I+\mathbb{H}_{1}\right\}
=\det\left\{  I+\mathbb{H}\left(  \varphi_{x,t}\right)  \right\}
.\label{lim 1}%
\end{equation}
Turn now to the second determinant in (\ref{det decom 0}). It is clearly a
$3\times3$ determinant. We are going to show that, in fact, this determinant
vanishes as $O\left(  \varepsilon\right)  $. To this end, we explicitly
evaluate it in the basis $\left(  k_{z_{n}}\right)  $%
\begin{align}
&  \det\left\{  I+\mathbb{H}_{0}-\mathbb{H}_{01}^{\ast}\left(  I+\mathbb{H}%
_{1}\right)  ^{-1}\mathbb{H}_{01}\right\}  \label{Det}\\
&  =\det\left(
\begin{array}
[c]{ccc}%
1+h_{-1-1}+\overline{d_{11}} & h_{-10}+d_{-10} & h_{-11}+d_{-11}\\
h_{0-1}+d_{0-1} & 1+h_{00}+d_{00} & h_{01}+d_{01}\\
\overline{h_{-11}}+\overline{d_{-11}} & h_{10}+d_{10} & 1+\overline{h_{-1-1}%
}+d_{11}%
\end{array}
\right)  ,\nonumber
\end{align}
where $h_{mn}$ and $d_{mn}$ are the matrix entries of%
\[%
%TCIMACRO{\dsum _{-1\leq n\leq1}}%
%BeginExpansion
{\displaystyle\sum_{-1\leq n\leq1}}
%EndExpansion
i\operatorname*{Res}\left(  \xi_{x,t}R_{\varepsilon},z_{-n}\right)
\left\langle \cdot,k_{z_{-n}}\right\rangle k_{z_{n}}%
\]
and%
\[
\mathbb{P}_{B}\mathbb{H}\left(  \Phi_{x,t}^{\varepsilon}\right)
\mathbb{P}_{B}-\mathbb{H}_{01}^{\ast}\left(  I+\mathbb{H}_{1}\right)
^{-1}\mathbb{H}_{01}%
\]
respectively. By Lemma \ref{Blaschke lemma}%
\begin{align}
h_{mn} &  =\left\langle
%TCIMACRO{\dsum _{-1\leq j\leq1}}%
%BeginExpansion
{\displaystyle\sum_{-1\leq j\leq1}}
%EndExpansion
i\operatorname*{Res}\left(  \xi_{x,t}R_{\varepsilon},z_{-j}\right)
\left\langle k_{z_{n}},k_{z_{-j}}\right\rangle k_{z_{j}},k_{z_{m}}^{\bot
}\right\rangle \label{h_mn}\\
&  =i\operatorname*{Res}\left(  \xi_{x,t}R_{\varepsilon},z_{-m}\right)
\left\langle k_{z_{n}},k_{z_{-m}}\right\rangle =i\operatorname*{Res}\left(
\xi_{x,t}R_{\varepsilon},z_{-m}\right)  k_{z_{n}}\left(  z_{-m}\right)
\nonumber\\
&  =\frac{\xi_{x,t}\left(  z_{-m}\right)  \operatorname*{Res}\left(
R_{\varepsilon},z_{-m}\right)  }{\overline{z_{m}}+\overline{z_{n}}}.\nonumber
\end{align}
Incidentally, (\ref{h_mn}) implies $h_{1-1}=\overline{h_{-11}},\ h_{-1-1}%
=\overline{h_{11}}$. Recall that $z_{n}$ are chosen so that $P\left(
z_{n}\right)  -i\varepsilon=0$ if $n=\pm1$ and $P\left(  z_{n}\right)
+i\varepsilon=0$ if $n=0$. Rewriting (\ref{R_epsilon}) as%
\[
R_{\varepsilon}\left(  k\right)  =aR\left(  k\right)  \frac{P\left(  k\right)
+2i\rho}{P\left(  k\right)  +i\rho\left(  1+a\right)  }\frac{P\left(
k\right)  }{P\left(  k\right)  +i\rho\left(  1-a\right)  }\frac{1}%
{B^{\varepsilon}\left(  k\right)  },
\]
for the residues we then have
\begin{align*}
&  \operatorname*{Res}\left(  R_{\varepsilon},z_{n}\right)  \\
&  =aR\left(  z_{n}\right)  \frac{P\left(  z_{n}\right)  +2i\rho}{P\left(
z_{n}\right)  +i\rho\left(  1+a\right)  }\frac{P\left(  z_{n}\right)
}{P\left(  z_{n}\right)  +i\rho\left(  1-a\right)  }\frac{2i\operatorname{Im}%
z_{n}}{B_{n}^{\varepsilon}\left(  z_{n}\right)  }.
\end{align*}
One now readily verifies that%
\begin{align*}
\frac{P\left(  z_{n}\right)  +2i\rho}{P\left(  z_{n}\right)  +i\rho\left(
1+a\right)  } &  =1+O\left(  \varepsilon^{2}\right)  ,\ \\
\frac{P\left(  z_{n}\right)  }{P\left(  z_{n}\right)  +i\rho\left(
1-a\right)  } &  =1+\left(  -1\right)  ^{n}\frac{\varepsilon}{2\rho}+O\left(
\varepsilon^{2}\right)  ,\\
B_{n}^{\varepsilon}\left(  z_{n}\right)  ^{-1} &  =1+5in\varepsilon/2+O\left(
\varepsilon^{2}\right)  ,
\end{align*}
and thus%
\begin{align}
&  \operatorname*{Res}\left(  R_{\varepsilon},z_{n}\right)  \label{Res}\\
&  =2i\operatorname{Im}z_{n}\ R\left(  z_{n}\right)  \left[  1+\frac
{i\varepsilon}{2}\left(  5n+\left(  -1\right)  ^{n}\frac{1}{\rho}\right)
+O\left(  \varepsilon^{2}\right)  \right]  .\nonumber
\end{align}
Inserting (\ref{Res}) into (\ref{h_mn}) yields%
\[
h_{mn}=\frac{2i\operatorname{Im}z_{m}}{\overline{z_{m}}+\overline{z_{n}}%
}\ \left(  \xi_{x,t}R\right)  \left(  z_{-m}\right)  \left[  1+\frac
{i\varepsilon}{2}\left(  5m+\left(  -1\right)  ^{m}\frac{1}{\rho}\right)
+O\left(  \varepsilon^{2}\right)  \right]  .
\]
Observe, that $h_{n,m}=O\left(  \varepsilon\right)  $ if $n\not =-m$ and
$h_{m,-m}$ doesn't vanish as $\varepsilon\rightarrow0$ (which is an important
fact for what follows). As we will see, only $h_{-11}$ and $h_{11}$ matter.
Recalling that $R\left(  1\right)  =-1$ we have%
\begin{align}
&  h_{-11}\label{h_-11}\\
&  =\xi_{x,t}\left(  1\right)  \left\{  1-\frac{\varepsilon}{2}\left[
\frac{1}{\rho}+i\overline{\xi_{x,t}\left(  1\right)  }\left(  R\xi
_{x,t}\right)  ^{\prime}\left(  1\right)  +5i\right]  +O\left(  \varepsilon
^{2}\right)  \right\}  ,\nonumber
\end{align}%
\begin{equation}
h_{-1-1}=\frac{i\varepsilon}{2}\ \xi_{x,t}\left(  1\right)  \left[  1+O\left(
\varepsilon\right)  \right]  .\label{h_-1-1}%
\end{equation}
Similarly, for the matrix $\left(  d_{mn}\right)  $ we have
\begin{align}
d_{mn} &  =\left\langle \mathbb{H}\left(  \Phi_{x,t}^{\varepsilon}\right)
k_{z_{n}},k_{z_{m}}^{\bot}\right\rangle -\left\langle \left(  I+\ \mathbb{H}%
_{1}\right)  ^{-1}\mathbb{H}_{01}k_{z_{n}},\mathbb{H}_{01}k_{z_{m}}^{\bot
}\right\rangle \label{big d}\\
&  =\frac{2\operatorname{Im}z_{m}}{\overline{B_{m}^{\varepsilon}\left(
z_{m}\right)  }}\left\{  \left\langle \mathbb{H}\left(  \Phi_{x,t}%
^{\varepsilon}\right)  k_{z_{n}},B_{m}^{\varepsilon}k_{z_{m}}\right\rangle
\right.  \nonumber\\
&  -\left.  \left\langle \left(  I+\mathbb{H}_{1}\right)  ^{-1}\mathbb{H}%
_{01}k_{z_{n}},\mathbb{H}_{01}B_{m}^{\varepsilon}k_{z_{m}}\right\rangle
\right\}  .\nonumber\\
&  =\varepsilon D_{mn}+O\left(  \varepsilon\right)  ,\nonumber
\end{align}
where $D_{mn}$ will be computed later. For the determinant in (\ref{Det}) we
clearly have%
\begin{align}
&  \det\left\{  I+\mathbb{H}_{0}-\mathbb{H}_{01}^{\ast}\left(  I+\mathbb{H}%
_{1}\right)  ^{-1}\mathbb{H}_{01}\right\}  \label{Det 1}\\
&  =\left(  1+h_{00}+d_{00}\right)  \det\left(
\begin{array}
[c]{cc}%
1+h_{-1-1}+\overline{d_{11}} & h_{-11}+d_{-11}\\
\overline{h_{-11}}+\overline{d_{-11}} & 1+\overline{h_{-1-1}}+d_{11}%
\end{array}
\right)  +O\left(  \varepsilon^{2}\right)  \nonumber\\
&  =2\left\{  \left\vert 1+h_{-1-1}+\overline{d_{11}}\right\vert
^{2}-\left\vert h_{-11}+d_{-11}\right\vert ^{2}\right\}  +O\left(
\varepsilon^{2}\right)  \text{ (by (\ref{h_-1-1})-(\ref{big d}))}\nonumber\\
&  =2\left(  1-\left\vert h_{-11}\right\vert ^{2}+2\operatorname{Re}%
h_{-1-1}\right)  +2\varepsilon\operatorname{Re}\left[  \overline{D_{11}}%
-\xi_{x,t}\left(  1\right)  \overline{D_{-11}}\right]  +O\left(
\varepsilon^{2}\right)  .\nonumber
\end{align}
Evaluate each term in the right hand side of (\ref{Det 1}) separately. By
(\ref{h_-11})-(\ref{h_-1-1}) one has%
\begin{align}
&  1-\left\vert h_{-11}\right\vert ^{2}+2\operatorname{Re}h_{-1-1}%
\label{h part}\\
&  =\varepsilon\left\{  1/\rho+\operatorname{Re}i\overline{\xi_{x,t}\left(
1\right)  }\left[  \left(  R\xi_{x,t}\right)  ^{\prime}\left(  1\right)
-1\right]  +O\left(  \varepsilon\right)  \right\}  \nonumber\\
&  =\frac{2\varepsilon}{\rho}\left\{  1+\rho\left(  x+12t\right)  -\frac{\rho
}{2}\sin\left(  2x+8t\right)  +O\left(  \varepsilon\right)  \right\}
\nonumber
\end{align}
and%
\begin{align*}
D_{mn} &  =\lim_{\varepsilon\rightarrow0}\left\{  \left\langle \mathbb{H}%
\left(  \Phi_{x,t}^{\varepsilon}\right)  k_{z_{n}},B_{m}^{\varepsilon}%
k_{z_{m}}\right\rangle -\left\langle \left(  I+\ \mathbb{H}_{1}\right)
^{-1}\mathbb{H}_{01}k_{z_{n}},\mathbb{H}_{01}B_{m}^{\varepsilon}k_{z_{m}%
}\right\rangle \right\}  \\
&  =:D_{mn}^{\left(  1\right)  }+D_{mn}^{\left(  2\right)  }.
\end{align*}
Since $\mathbb{H}\left(  \Phi_{x,t}^{\varepsilon}\right)  $ is a self-adjoint
operator, by Corollary \ref{smooth elements} we have ($m=\pm1,n=1$)%
\begin{align}
\overline{D_{mn}^{\left(  1\right)  }} &  =\lim_{\varepsilon\rightarrow
0}\left\langle \mathbb{H}\left(  \Phi_{x,t}^{\varepsilon}\right)
B_{m}k_{z_{m}},k_{z_{n}}\right\rangle \label{d1}\\
&  =\lim_{\varepsilon\rightarrow0}\left.  \mathbb{H}\left(  \Phi
_{x,t}^{\varepsilon}\right)  B_{m}k_{z_{m}}\right\vert _{z_{n}}\text{ \ \ (by
(\ref{rep kernels}))}\nonumber\\
&  =K_{m}\left(  n\right)  ,\nonumber
\end{align}
where%
\[
K_{m}\left(  n\right)  =-\int_{\mathbb{R}+ih}\frac{\varphi_{x,t}\left(
z\right)  }{\left(  z-m\right)  \left(  z+n\right)  }\frac{dz}{2\pi}.
\]
Similarly, by (\ref{conv of proj}), (\ref{conv in trace}), and Corollary
\ref{smooth elements} we have%
\begin{equation}
\overline{D_{mn}^{\left(  2\right)  }}=-\left.  \mathbb{H}\left(
\varphi_{x,t}\right)  \left(  I+\mathbb{H}\left(  \varphi_{x,t}\right)
\right)  ^{-1}K_{m}\right\vert _{n+i0}.\label{d2}%
\end{equation}
Therefore, combining (\ref{d1}) and (\ref{d2}) we have%
\begin{align*}
\overline{D_{mn}} &  =K_{m}\left(  n\right)  -\left.  \mathbb{H}\left(
\varphi_{x,t}\right)  \left(  I+\mathbb{H}\left(  \varphi_{x,t}\right)
\right)  ^{-1}K_{m}\right\vert _{n+i0}\\
&  =\left.  \left(  I+\mathbb{H}\left(  \varphi_{x,t}\right)  \right)
^{-1}K_{m}\right\vert _{n+i0}.
\end{align*}
Substituting this and (\ref{h part}) into (\ref{Det 1}) yields%
\begin{align}
&  \frac{\rho}{4\varepsilon}\det\left\{  I+\mathbb{H}_{0}-\mathbb{H}%
_{01}^{\ast}\left(  I+\mathbb{H}_{1}\right)  ^{-1}\mathbb{H}_{01}\right\}
\nonumber\\
&  =1+\rho\left(  x+12t\right)  -\sin\left(  2x+8t\right)  \nonumber\\
&  +\frac{\rho}{2}\operatorname{Re}\left.  \left(  I+\mathbb{H}\left(
\varphi_{x,t}\right)  \right)  ^{-1}\left(  K_{1}-\xi_{x,t}\left(  1\right)
K_{-1}\right)  \right\vert _{1+i0}+O\left(  \varepsilon\right)
\label{piece of det}%
\end{align}
We have now prepared all the ingredients to find the solution to the KdV
equation with the initial data $Q_{\varepsilon}$ by the Dyson formula. Indeed,%
\begin{align}
Q_{\varepsilon}\left(  x,t\right)   &  =-2\partial_{x}^{2}\log\det\left\{
I+\mathbb{H}\left(  \varphi_{x,t}^{\varepsilon}\right)  \right\}  \text{ \ (by
(\ref{Det 1}))}\label{Q_epsilon}\\
&  =2\partial_{x}^{2}\log\det\left(  I+\mathbb{H}_{1}\right)  \nonumber\\
&  -2\partial_{x}^{2}\log\det\left\{  I+\mathbb{H}_{0}-\mathbb{H}_{01}^{\ast
}\left(  I+\mathbb{H}_{1}\right)  ^{-1}\mathbb{H}_{01}\right\}  \nonumber\\
&  =-2\partial_{x}^{2}\log\det\left\{  I+\mathbb{H}\left(  \varphi
_{x,t}\right)  \right\}  \text{ \ (by (\ref{lim 1}) and (\ref{piece of det}%
))}\nonumber\\
&  -2\partial_{x}^{2}\log\left\{  1+\rho\left(  x+12t\right)  -\frac{\rho}%
{2}\sin\left(  2x+8t\right)  \right.  \nonumber\\
&  +\left.  \frac{\rho}{2}\operatorname{Re}\left.  \left(  I+\mathbb{H}\left(
\varphi_{x,t}\right)  \right)  ^{-1}\left(  K_{1}-\xi_{x,t}\left(  1\right)
K_{-1}\right)  \right\vert _{1+i0}\right\}  +O\left(  \varepsilon\right)
.\nonumber
\end{align}
We are now able to fill the gap left in the proof of Theorem
\ref{Thm on approx}, i.e. (\ref{lim pot}). To this end, set $t=0$ in
(\ref{Q_epsilon}) and take $x>0$. In this case $\xi_{x,0}\in H^{\infty}$ and
hence $\varphi_{x,0}\in H^{\infty}$. Therefore, $\mathbb{H}\left(
\varphi_{x,t}\right)  =0$ and by the Lebesgue dominated convergence theorem
(or by Corollary \ref{smooth elements}) we also have%
\begin{align*}
K_{m}\left(  s\right)   &  =-\int_{\mathbb{R}+ih}\frac{\varphi_{x,0}\left(
z\right)  }{\left(  z-m\right)  \left(  z+s\right)  }\frac{dz}{2\pi}\\
&  =-\lim_{h\rightarrow\infty}\int_{\mathbb{R}+ih}\frac{\varphi_{x,0}\left(
z\right)  }{\left(  z-m\right)  \left(  z+s\right)  }\frac{dz}{2\pi}=0.
\end{align*}
Eq. (\ref{Q_epsilon}) simplifies now to read%
\[
Q_{\varepsilon}\left(  x,0\right)  =-2\partial_{x}^{2}\log\left(  1+\rho
x-\frac{\rho}{2}\sin2x\right)  +O\left(  \varepsilon\right)  ,\ \ \ x>0.
\]
Recalling (\ref{Q}), we conclude that $Q_{\varepsilon}\left(  x\right)
=Q_{\varepsilon}\left(  x,0\right)  \rightarrow q_{0}\left(  x\right)  $ for
$x>0$. Since $Q_{\varepsilon}\left(  x\right)  $ is even, (\ref{lim pot}) follows.

Pass now in (\ref{Q_epsilon}) to the limit as $\varepsilon\rightarrow0$.
Apparently,
\begin{align*}
\lim_{\varepsilon\rightarrow0}Q_{\varepsilon}\left(  x,t\right)   &
=-2\partial_{x}^{2}\log\det\left\{  I+\mathbb{H}\left(  \varphi_{x,t}\right)
\right\}  \text{ \ }\\
&  -2\partial_{x}^{2}\log\left\{  1+\rho\left(  x+12t\right)  -\frac{\rho}%
{2}\sin\left(  2x+8t\right)  \right. \\
&  +\left.  \frac{\rho}{2}\operatorname{Re}\left.  \left(  I+\mathbb{H}\left(
\varphi_{x,t}\right)  \right)  ^{-1}\left(  K_{1}-\xi_{x,t}\left(  1\right)
K_{-1}\right)  \right\vert _{1+i0}\right\}  .
\end{align*}
By the Bourgain theorem $Q\left(  x,t\right)  =\lim_{\varepsilon\rightarrow
0}Q_{\varepsilon}\left(  x,t\right)  $ is the (unique) solution to the KdV
equations with data $Q\left(  x\right)  $. Recalling Corollary
\ref{smooth elements}, we see that
\[
K_{n}=\mathbb{H}(\varphi_{x,t})k_{n+i0},\ \ \ n=\pm1.
\]
This completes the proof of the theorem.
\end{proof}

Note that the first term $u_{0}\left(  x,t\right)  $ in the solution (\ref{u})
is given by the same Dyson formula (\ref{u0}) as in the short-range case but
of course $u_{0}\left(  x,0\right)  $ is not a short range potential. The
second term $u_{1}\left(  x,t\right)  $ in (\ref{u}) is responsible for the
bound state $+1$ and if $\rho=1$ it resembles the so-called positon solution%
\begin{equation}
u_{\text{pos}}\left(  x,t\right)  =-2\partial_{x}^{2}\log\left\{
1+x+12t-\frac{1}{2}\sin2\left(  x+4t\right)  \right\}  .
\label{tau for positon}%
\end{equation}
Such solutions seem to have appeared first in the late 70s earlier 80s but a
systematic approach was developed a decade later by V. Matveev (see his 2002
survey \cite{Mat02}).

The formula (\ref{tau for positon}) readily yields basic properties of
one-position solutions. (1) As a function of the spatial variable
$u_{\text{pos}}\left(  x,t\right)  $ has a double pole real singularity which
oscillates in the $1/2$ neighborhood of the moving point $x=-12t-1$. (2) For a
fixed $t\geq0$%
\begin{equation}
u_{\text{pos}}\left(  x,t\right)  =-4\frac{\sin2\left(  x+4t\right)  }%
{x}+O\left(  x^{-2}\right)  ,\ \ \ x\rightarrow\pm\infty\text{.}%
\label{posit asym}%
\end{equation}
Observe that%
\[
u_{\text{pos}}\left(  x,0\right)  =-2\partial_{x}^{2}\log\left(  1+x-\frac
{1}{2}\sin2x\right)  ,
\]
which coincides on $\left(  0,\infty\right)  $ with our $Q\left(  x\right)  $
for $\rho=1$. Moreover, comparing (\ref{Q asym}) with (\ref{posit asym}) one
can see that the asymptotic behaviors for $x\rightarrow-\infty$ of our
$Q\left(  x\right)  $ with $\rho=1$ and $u_{\text{pos}}\left(  x,0\right)  $
differ only by $O\left(  x^{-2}\right)  $. But, of course, $Q\left(  x\right)
$ is bounded on $\left(  -\infty,0\right)  $ while $u_{\text{pos}}\left(
x,0\right)  $ is not. Note also that the positon is somewhat similar to the
soliton given by%
\begin{equation}
u_{\text{sol}}\left(  x,t\right)  =-2\partial_{x}^{2}\log\cosh\left(
x-4t\right)  .\label{tau for soliton}%
\end{equation}
As opposed to the soliton, the positon has a square singularity (not a smooth
hump) moving in the opposite direction three times as fast.

We note that multi-positon as well as soliton-positon solutions have been
studied in great detail (see \cite{Mat02} the references cited therein). In
\cite{Mat02} Matveev also raises the equation if there is a bounded positon,
i.e. a solution having all properties of a positon but is regular. We are
unable to tell if our solution is a bounded positon or not.


\begin{thebibliography}{99}                                                                                               %


\bibitem {AC91}Ablowitz, M. J.; Clarkson , P. A. \emph{Solitons, nonlinear
evolution equations and inverse scattering}. London Mathematical Society
Lecture Note Series, 149. Cambridge University Press, Cambridge, 1991. xii+516 pp.

\bibitem {ADM81}Abraham, P. B.; DeFacio, B.; Moses, H. E. Two distinct local
potentials with no bound states can have the same scattering operator: a
nonuniqueness in inverse spectral transformations. Phys. Rev. Lett. 46 (1981),
no. 26, 1657--1659.

\bibitem {AK01}Aktosun,T. and Klaus M.. Chapter 2.2.4:\textit{\ }\emph{Inverse
theory: problem on the line}. In: E. R. Pike and P. C. Sabatier (eds.),
Scattering, Academic Press, London, 2001, pp. 770.

\bibitem {TeschlRarefaction16}K. Andreiev, I. Egorova, T.-L. Lange, G. Teschl
\emph{Rarefaction waves of the Korteweg-de Vries equation via nonlinear
steepest descent}, J. Differential Equations 261 (2016), 5371-5410.

\bibitem {BotSil02}B\"{o}tcher, A.; Silbermann B. \emph{Analysis of Toeplitz
operators}. Springer-Verlag, Berlin, 2002. 665 pp.

\bibitem {Bourgain93}Bourgain, J.\ \emph{Fourier transform restriction
phenomena for certain lattice subsets and applications to nonlinear evolution
equations I, II.} Geom. Funct. Anal., 3:107--156 (1993), 209--262.

\bibitem {MatveevOpenProblems}Dubard, P.; Gaillard, P.; Klein, C.; and
Matveev, V.B. \emph{On multi-rogue wave solutions of the NLS equation and
positon solutions of the KdV equation.} Eur. Phys. J. Special Topics 185
(2010), 247--258.

\bibitem {DubMatNov76}Dubrovin, B. A.; Matveev, V. B.; Novikov, S. P.
\emph{Nonlinear equations of Korteweg-de Vries type, finite-band linear
operators and Abelian varieties}, (Russian) Uspehi Mat. Nauk 31 (1976), no. 1
(187), 55--136.

\bibitem {Egorovaetal13}Egorova, Iryna; Gladka, Zoya; Kotlyarov, Volodymyr;
Teschl, Gerald \emph{Long-time asymptotics for the Korteweg--de Vries equation
with step-like initial data.} Nonlinearity 26 (2013), no. 7, 1839--1864.

\bibitem {TeschlShock2016}Egorova, I.; Gladka, Z.; and Teschl, G. \emph{On the
form of dispersive shock waves of the Korteweg-de Vries equation,} Zh. Mat.
Fiz. Anal. Geom. 12 (2016), 3-16.

\bibitem {GGKM67}Gardner, C. S.; Greene, J. M.; Kruskal, M. D.; and Miura, R.
M Phys. Rev. Lett. 19 (1967), 1095--1097.

\bibitem {GesHold03}Gesztesy, Fritz; Holden, Helge \emph{Soliton equations and
their algebro-geometric solutions. Vol. I. (1+1)-dimensional continuous
models}. Cambridge Studies in Advanced Mathematics, 79. Cambridge University
Press, Cambridge, 2003. xii+505 pp.

\bibitem {GruRybSIMA15}Grudsky, S. and Rybkin, A. \emph{Soliton theory and
Hakel operators}, SIAM\ J. Math. Anal., 47 (2015) no 3, 2283-2323.

\bibitem {GT09}Grunert, Katrin; Teschl, Gerald \emph{Long-time asymptotics for
the Korteweg-de Vries equation via nonlinear steepest descent}. Math. Phys.
Anal. Geom. 12 (2009), no. 3, 287--324.

\bibitem {Gurevich73}Gurevich, A.V. and Pitaevskii, P. \emph{Decay of initial
discontinuity in the Korteweg--de Vries equation}, JETP Letters \textbf{17:5}
(1973), 193--195.

\bibitem {HintonKlausIP1989}Hinton, D. B.; Klaus, M.; Shaw, J. K.
\emph{High-energy asymptotics for the scattering matrix on the line. }Inverse
Problems 5 (1989), no. 6, 1049--1056.$\;$

\bibitem {Hirota71}Hirota, R. \emph{Exact solution of the Korteweg de Vries
equation for multiple collisions of solitons}. Phys. Rev. Lett. 27 (1971), 1192--1194.

\bibitem {Hruslov76}Hruslov, \={E}. Ja. \emph{Asymptotic behavior of the
solution of the Cauchy problem for the Korteweg-de Vries equation with
steplike initial data}. Mat. Sb. (N.S.) 99(141) (1976), no. 2, 261--281, 296.

\bibitem {ItsMat75}Its, A. R.; Matveev, V. B. \emph{Schr\"{o}dinger operators
with the finite-band spectrum and the N-soliton solutions of the Korteweg-de
Vries equation.} (Russian) Teoret. Mat. Fiz. 23 (1975), no. 1, 51--68. English
translation: Theoret. and Math. Phys. 23 (1975), no. 1, 343--355 (1976).

\bibitem {Klaus82}Klaus, Martin. \emph{Some applications of the
Birman-Schwinger principle}. Helv. Phys. Acta 55 (1982/83), no. 1, 49--68.

\bibitem {Klaus91}Klaus, Martin. \emph{Asymptotic behavior of Jost functions
near resonance points for Wigner-von Neumann type potentials.} J. Math. Phys.
32 (1991), no. 1, 163--174.

\bibitem {LevitanInverseProblems87}Levitan, B. M. \emph{Inverse
Sturm-Liouville problems.} Translated from the Russian by O. Efimov. VSP,
Zeist, 1987. x+240 pp.

\bibitem {MarchBook2011}Marchenko, Vladimir A. \emph{Sturm-Liouville operators
and applications}. Revised edition. AMS Chelsea Publishing, Providence, RI,
2011. xiv+396 pp.

\bibitem {Mat02}Matveev, V. B. \emph{Positons: slowly decreasing analogues of
solitons.} Theor. Math. Phys. 131 (2002), no. 1, 483--497.

\bibitem {WvN1929}J. von Neumann and E. P. Wigner, \emph{Uber merkwu rdige
diskrete Eigenwerte}, Phys. Z. 30 (1929), 465 467.

\bibitem {Nik2002}Nikolski, N. K. \emph{Operators, functions, and systems: An
easy reading. Volume 1: Hardy, Hankel and Toeplitz.} Mathematical Surveys and
Monographs, vol. 92, Amer. Math. Soc., Providence, 2002. 461 pp.

\bibitem {NikBook86}Nikolski, N. K. \emph{Treatise on the shift operator.
Spectral function theory}. With an appendix by S. V. Hru\v{s}\v{c}ev [S. V.
Khrushch\"{e}v] and V. V. Peller. Translated from the Russian by Jaak Peetre.
Grundlehren der Mathematischen Wissenschaften [Fundamental Principles of
Mathematical Sciences], 273. Springer-Verlag, Berlin, 1986. xii+491 pp.

\bibitem {NovikovKhenkin84}Novikov, R. G.; Khenkin, G. M. \emph{Oscillating
weakly localized solutions of the Korteweg-de Vries equation.} (Russian)
Teoret. Mat. Fiz. 61 (1984), no. 2, 199--213.

\bibitem {NPZ}Novikov, S. P.; Manakov, S. V.; Pitaevski{\u{\i}} L. P.;
Zakharov, V. E. \emph{Theory of solitons}, Contemporary Soviet Mathematics
(1984), Consultants Bureau [Plenum], New York, The inverse scattering method,
Translated from Russian.

\bibitem {Peller2003}Peller, V. V. \emph{Hankel operators and their
applications. Springer Monographs in Mathematics}. Springer-Verlag, New York,
2003. xvi+784 pp. ISBN: 0-387-95548-8.

\bibitem {RybJMP08}Rybkin, Alexei. \emph{On the evolution of a reflection
coefficient under the KdV flow}. Journal of Mathematical Physics, 49 (2008) 072701.

\bibitem {RybPhysD2018}Rybkin, Alexei. \emph{KdV equation beyond standard
assumption on initial data}. Physica D: Nonlinear Phenomena, 365 (2018), 1--11.

\bibitem {Titchmarsh62}Titchmarsh, E. C. \ \emph{Eigenfunction expansions
associated with second-order differential equations. Part I.} Second Edition
Clarendon Press, Oxford 1962 vi+203 pp.
\end{thebibliography}
\end{document}